\documentclass[11pt]{article}

\RequirePackage[T1]{fontenc}
\RequirePackage[latin2]{inputenc}
\RequirePackage{times}

\usepackage{graphicx}
\usepackage{url}
\usepackage{amsthm}
\usepackage{amsmath}
\usepackage{amssymb}

\newtheorem{theorem}{Theorem}
\newtheorem{corollary}[theorem]{Corollary}

\theoremstyle{definition}
\newtheorem{definition}[theorem]{Definition}
\newtheorem{assumption}[theorem]{Assumption}
\newtheorem{example}[theorem]{Example}

\title{The Power of Spreadsheet Computations}

\author{JERZY TYSZKIEWICZ\footnote{Research funded by Polish National Science Centre (Narodowe Centrum
Nauki)}\\Institute of Informatics, University
of Warsaw,\\ Banacha 2, 02-097 Warszawa, Poland\\\texttt{jty@mimuw.edu.pl}}

\begin{document}

\begin{titlepage}

\maketitle

\begin{abstract} We investigate the expressive power of
spreadsheets. We consider spreadsheets which contain only formulas,
and assume that they are small templates, which can be filled to a
larger area of the grid to process input data of variable
size. Therefore we can compare them to well-known machine models of
computation. We consider a number of classes of spreadsheets defined
by restrictions on their reference structure. Two of the classes
correspond closely to parallel complexity classes: we prove a direct
correspondence between the dimensions of the spreadsheet and amount of
hardware and time used by a parallel computer to compute the same
function. As a tool, we produce spreadsheets which are universal in
these classes, i.e. can emulate any other spreadsheet from them. In
other cases we implement in the spreadsheets in question instances of
a polynomial-time complete problem, which indicates that the the
spreadsheets are unlikely to have efficient parallel evaluation
algorithms. Thus we get a picture how the computational power of
spreadsheets depends on their dimensions and structure of references.
\end{abstract}

\end{titlepage}

\section{Introduction}\label{S1}

\subsection{Why spreadsheets?}

Spreadsheets are an extremely popular type of software systems. They
have conquered very diverse areas of present day politics, business,
research, and, last but not least, our private lives. However, this
prevalence is not so evident, because spreadsheets are typically used
in the back office and are not presented to the public. They make to
the news only when something went really wrong: for instance a
spreadsheet used to justify a widely implemented public policy, as in
the case of the extremely influential report \cite{Debt} concerning a
supposed causal relationship between high national debt and low
economic growth, turns out to contain an error in a formula, affecting
the outcome of the calculations \cite{NoDebt}. Or when an
\textit{Excel\/} model was used to manage the investments of JPMorgan
Chase \& Co. bank, which led to trade losses estimated in billions of
dollars \cite{JPMorgan}. Research is not an exception, and a careful
reader of \textit{Science\/} magazine can read \cite{science1,science2},
in which a scientific controversy finally turns out to be related to a
spreadsheet mistake.

Indeed, spreadsheets are among the most frequently used software tools
of any kind. Despite that, and surprisingly enough, very little has
been known about their computational power. Consequently, the users
are not really aware of the abilities and limitations of the tool they
use.

\subsection{How to measure expressiveness of spreadsheets}

The aim of this paper is to analyze the power of spreadsheets
considered as a tool for specifying general-purpose computations. 

It is clear that they belong to the nonuniform computation models,
where for each input size there is a separate computing
device. Uniformity can be introduced to such a model by imposing that
there is a common, low complexity procedure to create those devices,
given the input size. In this respect, spreadsheets come with a
natural, built-in solution of this issue: \textit{filling},
i.e. automated copying of formulas into neighboring cells, with
suitable reference adjustments. In spreadsheet programming practice,
filling is the usual way to produce a spreadsheet processing a large
amount of data from a few formulas prepared manually, or to adjust an
already existing one to accommodate a new supply of data.

We follow this idea and treat a small spreadsheet consisting of a few
initial cells with formulas as a program, and filling as a method to
run it on input data of variable size. The input data is also located
in the spreadsheet, but is not subject to filling.  For mathematical
convenience, we imagine that the spreadsheet grid is actually infinite
and any number of rows and columns can be filled with copies of the
initial cells. These two assumptions allow us to apply the methods of
computational complexity, which are asymptotic in nature, to the study
of spreadsheets.

No macros and user-defined functions written in a general programming
language are permitted. A few specific and infrequently occurring
spreadsheet functions are also excluded from our analysis.

\subsection{Main results}

The structural properties of spreadsheets which determine their
computational properties are defined by restrictions on the pattern of
references of formulas to other cells and do not apply to the
references to input data:

\begin{enumerate}
\item A spreadsheet is row-organized if all ranges to which its
  formulas refer are horizontal ones, i.e., rows.

\item A spreadsheet is row-directed if every formula refers only to
  cells and ranges located strictly above the location of the formula
  itself.
\end{enumerate}

These properties are preserved by filling. Our findings described
below refer also to the analogous column-organized and column-directed
spreadsheets. Thus we get several possible combinations of
column-related and row-related properties.

We describe the computational power of spreadsheets by relating them to
Parallel Random Access Machines (PRAM for short).

We encounter both Concurrent Read Concurrent Write (CRCW) priority
write machines and Exclusive Read Exclusive Write (EREW) ones.

On several occasions we proceed by implementing instances of the
$P$-complete \textit{Circuit Value Problem\/} (abbreviated CVP) in
spreadsheets, in order to demonstrate that they are unlikely to have
efficient parallel evaluation algorithms.

Our first main result is that any given initial row-organized
row-directed spreadsheet can be converted into a program $\pi$ for an
EREW PRAM such that the function computed by that spreadsheet filled
to the dimensions of $c$ columns and $r$ rows is always the same as
that computed by $\pi$ evaluated on a PRAM with $c$ processors, $O(c)$
cells of memory and running for $O(r\log c)$ time. Thus, if a
spreadsheet is row-organized row-directed, it can be efficiently
parallelized: its number of columns contributes only a logarithmic
factor to the total computation time. This sets an upper bound on the
computational power of row-organized row-directed spreadsheets. An
analogous result holds for column-organized column-directed
spreadsheets.

In order to get a lower bound, and thus determine the class of
functions computable by those spreadsheets, we prove our second main
result: there is a row-organized row-directed spreadsheet with 19
formulas which is a universal CRCW PRAM evaluator, i.e., one which
given a (suitably encoded) program $\pi$ together with its input, and
filled to the dimensions of $p$ columns and $10t$ rows, computes in
its last row the description of PRAM after executing $t$ steps of
$\pi$ on $p$ processors and with $p$ cells of shared memory. This
demonstrates that spreadsheets can implement a natural and broad class
of general-purpose computations.

This spreadsheet is also a universal row-organized row-directed
spreadsheet: any other spreadsheet from this class can be equivalently
expressed as a program for a PRAM, which in turn can be executed on
that spreadsheet. This above results demonstrate that row-organized
row-directed spreadsheets and PRAMs are almost equivalent in computing
power, with clear relations between the resources in both
models. Indeed, translating a spreadsheet into an equivalent program
for PRAM, and then back to spreadsheet, incurs only a logarithmic
overhead, a common effect of translations between different parallel
computation models.

\begin{figure}
\includegraphics[width=5in]{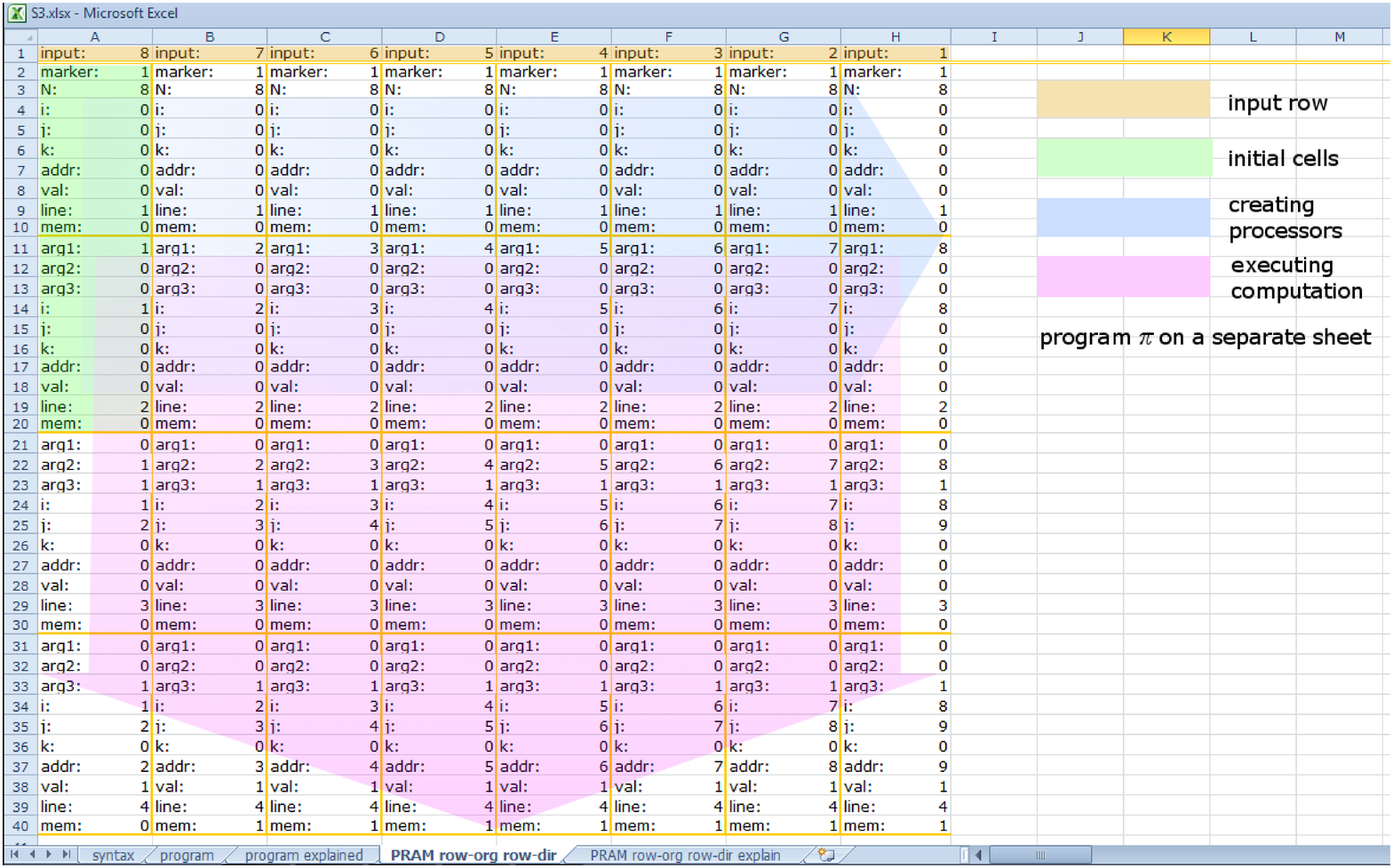}
\caption{PRAM evaluator in a row-organized row-directed spreadsheet
  S1: structure and mode of operation. Processors are located in
  columns, and computation time advances downwards. A vertical group
  of 10 cells constitutes a snapshot of a processor at a given time,
  so that extending computation time by one unit requires filling 10
  rows. It is provided as an electronic appendix S1.}
\label{F1}
\end{figure}

At the same time PRAMs and spreadsheets are extremely different: the
former have only programming primitives and no data analysis ones,
while the latter have only data analysis functions and do not support
any form of programming on the level of the spreadsheet itself.

Our result indicates that it is possible and feasible to turn
row-organized row-directed spreadsheets into a potential programming
language for specifying general-purpose parallel computations, with
clear relation between the dimensions of the spreadsheet and
complexity of its formulas, and the computational resources necessary
to execute it. They might be particularly dedicated for users who have
high computing needs but no programming training.

Let us move to other classes of spreadsheets, not necessarily
row-oriented ones. We demonstrate a row-organized but not row-oriented
PRAM simulator (electronic appendix S2), more powerful than the one previously
described. If it is extended to $c$ columns and $r$ rows, computes the
description of PRAM after executing $ cr/10p$ steps of $\pi$ on $p$
processors and with $p$ cells of shared memory, where $ p\le c$ is a
part of the input. Thus this PRAM simulator can perform a parallel or
a sequential computation, depending on the actual input, trading the
number of processors and cells of memory for more computation
time. 

To get the results for other classes of spreadsheets, we implement
instances of CVP in them. Each time we do so, we get hypothetical
lower bounds on the parallel complexity of evaluating spreadsheets in
this class. The larger instance of this problem can be implemented,
the higher is the lower bound.

\begin{table}
\label{Tab1}}%
{\begin{tabular}{|p{1in}|p{2in}|p{2in}|} 
\hline 
& Column directed & Column un-directed \\ 
\hline
Row directed &
\begin{itemize}
\item Upper bounds $O(r\log cr)$ with $c$ processors
and $ O(c\log cr)$ with $r$ processors on EREW PRAM
\item  No known PRAM
simulation
\item $2n$ columns and $ 4n$ rows implement CVP
instance of size $n$ 
\end{itemize}
& 
\begin{itemize}
\item Upper bound $ O(r\log cr)$ with
$c$ processors on EREW PRAM
\item $c$ columns and $r$ rows simulate any CRCW PRAM with
$c$ processors and cells of memory for $ r/10$ steps 
\item 1 column and $ n$ rows implement CVP instance of size $n$ 
\end{itemize}
\\ \hline
Row un-directed & 
\begin{itemize}
\item Upper bound $ O(c\log cr)$ time on EREW PRAM with
$r$ processors
\item $c$ columns and $r$ rows simulate any PRAM with
$r$ processors and cells of memory for $ c/10$ steps
\item $ n$ columns
and 1 row implement any CVP instance of size $ n$ 
\end{itemize}
& 
\begin{itemize}
\item Upper bound
$ O(c^{2} r^{2} \log cr)$ with 1 processor
\item $c$ columns and $r$ rows simulate any CRCW PRAM with $p$
  processors and cells of memory for $ cr/10p$ steps, $p$ is a part of
  the input
\item $c$ columns and $r$ rows implement any CVP instance of size $
  cr/8$
\end{itemize}
\\ \hline
\end{tabular}
\caption{Summary of demonstrated upper and lower bounds for
spreadsheet computations, depending on their structure. We disregard
small changes depending on whether spreadsheets are row- or
column-organized or un-organized, which are discussed in the main
text.}
\end{table}

Three highlights from Table \ref{Tab1} are the following:

\begin{enumerate}
\item Row-oriented but not row-organized spreadsheets have parallel
  evaluation algorithms with $c$ processors and of time complexity
  $O(r\log cr)$, similar to those for row-organized row-oriented ones
  discussed above, except that they need much more memory: $O(cr)$
  instead of $O(c)$ cells.

\item For spreadsheets which are row-oriented but not column-oriented,
  and its dual class, one of the dimensions contributes a logarithmic
  factor to the computation time, while a CVP instance can be encoded
  in the other dimension, which causes its size to appear as a linear
  factor in the computation time.

\item For spreadsheets which are simultaneously row-oriented and
  column-oriented, one has choice which of the dimensions will
  contribute a logarithmic factor and which a linear one to the
  computation time. However, it is unlikely that there is an algorithm
  in which both dimensions contribute only logarithmic factors. This
  is demonstrated by electronic appendix S3, which implements CVP ``diagonally''
  in the spreadsheet. There is a more complicated, column-organized
  version of this implementation, too.
\end{enumerate}

All the above results taken together give a comprehensive picture of
the computing power of spreadsheets without macros and without a few
forbidden functions. It turns out that this power is strongly
influenced by the pattern of references within the spreadsheet, in
addition to its size. Moreover, two dual classes of spreadsheets
appear to be interesting candidates for a general purpose parallel
programming language, aimed at users without expertise in programming.

\subsection{Plan of the paper}

In the following three Sections \ref{S2}--\ref{S4} we introduce
\textit{Excel\/} spreadsheets, parallel random access machines and the
necessary notions of complexity theory.

Next in Section \ref{S5} we discuss spreadsheets of unrestricted
structure, demonstrating that they can compute $P$-complete problems.

Then in the following two Sections \ref{S6} and \ref{S7} we consider
row-directed and row-organized spreadsheets, demonstrating efficient
parallel algorithms to evaluate them.

Then in the next two Sections \ref{S8} and \ref{S9} we show that a
natural class of parallel computations is expressible in row-directed
and row-organized spreadsheets. We achieve that by proving a theorem
about simulation of CRCW PRAMs by row-organized row-oriented
restricted spreadsheets.

Then in Section\ref{S10} we advocate the use of such spreadsheets as a
parallel programming language.

In the following Section \ref{S11} we simulate PRAMs with
input-dependent number of processors by row-organized but not
row-oriented restricted spreadsheets. We also indicate why in this
case it is much harder to get a precise description of the computing
power of spreadsheets in terms of machine models computation.

In Section \ref{S12} we explain why several functionalities of
spreadsheets, including some specific functions and circular
computations, are excluded from our analysis.

In the last Section \ref{S13} we briefly indicate related research.

\section{Spreadsheets}\label{S2}

We assume the reader to be familiar with spreadsheets and their
use. Syntactically the present paper is based on \textit{Microsoft
  Excel\/}$^{TM}$ and the reference to syntax and meaning of formulas
is the on-line help of \textit{Microsoft Excel\/}
\cite{Excel-help}. The demonstration software accompanying this paper
is also prepared for \textit{Excel\/}. Our spreadsheets are provided
in the .xlsx format, but they are compatible with the earlier .xls
format and can be saved in it without any loss of functionality.

In this paper we do not use macros or user defined functions written
in Visual Basic or any other external programming language. So our
spreadsheets consist of formulas constructed from the basic, built-in
spreadsheet functions and nothing else.

In the spreadsheets provided with the paper we frequently use
names. This is a method to assign a name to a frequently used range of
cells, and later on use that name in formulas to denote that range. We
use it for the sole purpose of making formulas shorter and easier to
understand. This method does not increase the computational abilities
of \textit{Excel\/}.

One of the main features of spreadsheets we will use in this paper is
the fill operation, which is performed by selecting a rectangular
range of cells, clicking a small handle in the lower right corner of
it and extending its boundaries either horizontally or vertically,
which results in copying the content of the initial range to the new,
larger area of the worksheet. An example is shown in Figure 2.

\begin{figure}
\includegraphics[width=5in]{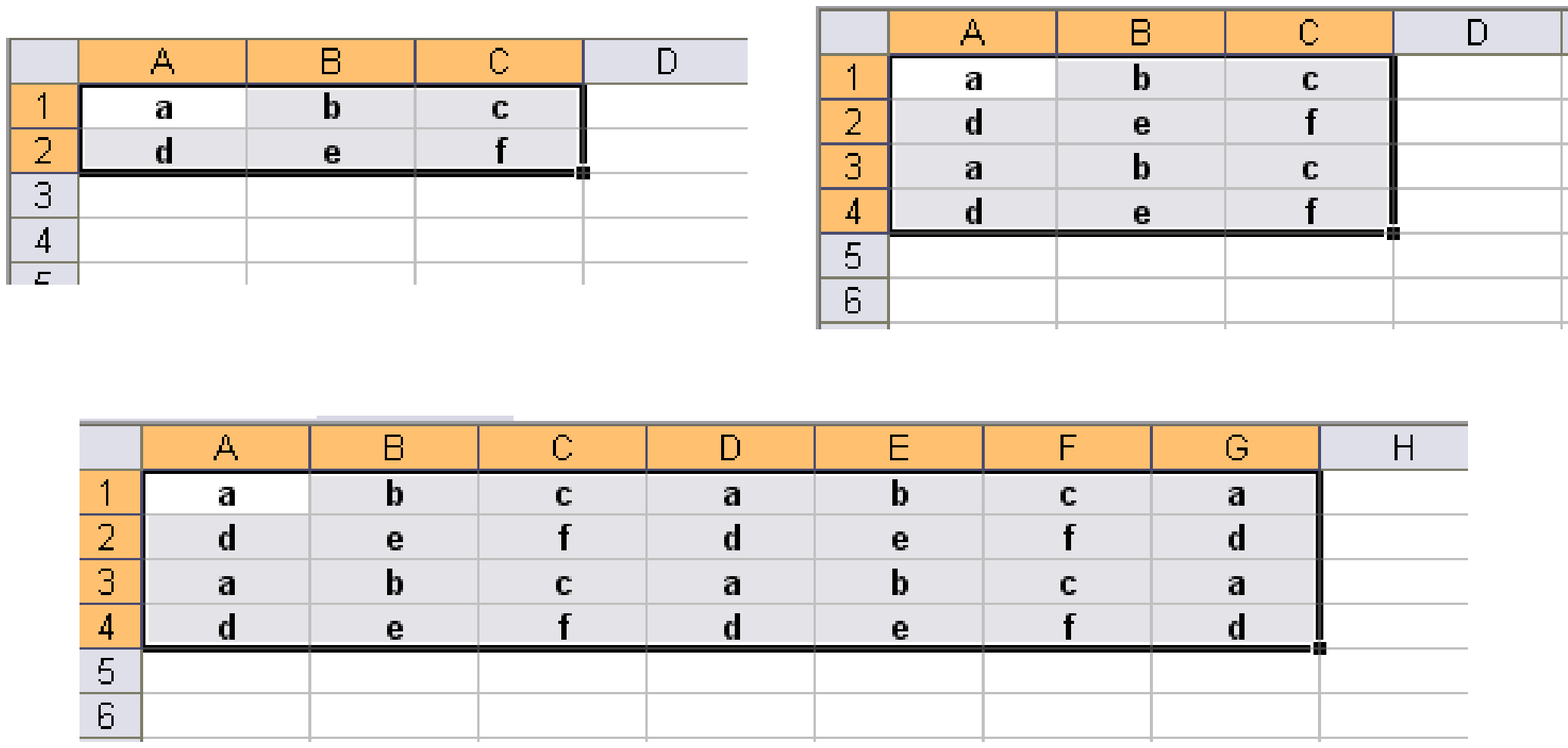}
\caption{The result of filling an initial range
twice: first by filling two additional rows, and then four columns.}
\label{F2}
\end{figure}

\begin{assumption}\label{A1} We consider the spreadsheet grid to be
infinite. Therefore, the user can start with a small spreadsheet and
fill any desired number $r$ of rows and $c$ of columns with copies of
the initial formulas. We assume also that there is no limit on the
numbers used in the spreadsheet, so that we can represent the location
of every cell in the spreadsheet by a number.
\end{assumption}

\begin{definition}\label{D1} The \textit{computing part of a
spreadsheet\/} $S$ is the fragment of it, used for filling the grid
  with formulas, in order to process the input data of variable
  size. The \textit{input part\/} of this spreadsheet are those cells,
  which are not subject to filling and do not contain formulas
  referring to the cells filled. The \textit{output part\/\/} of a
  spreadsheet is the last created row/column of cells.

We always assume that these three parts are indicated in the
spreadsheets we discuss below, although spreadsheets used in everyday
practice do not make this distinction clear.
\end{definition}

\begin{definition}\label{D2}\textit{Computation of a spreadsheet\/} is
  the act of marking the computing part of this spreadsheet and
  filling a certain number of rows and columns (in either order) by
  formulas from this part, to which the spreadsheet system
  (\textit{Excel\/} in this paper) responds by evaluating the so
  created formulas. In particular, this process produces the
  \textit{output\/}, which consists of the values computed in the
  output part.
\end{definition}

Under Assumption \ref{A1} and Definitions \ref{D1} and \ref{D2} we can
consider a spreadsheet as a program, with clearly identified areas for
input, executable part and output. The process of filling allows one
to use this program for processing input data of variable size. This
is a crucial step, which allows us to describe the computing power of
spreadsheets using common complexity classes.

Now we are going to define two crucial structural
characteristics of spreadsheets. It will turn out that they determine
their ability to express computations of certain complexity. Both
characteristics are preserved by filling.

\begin{definition}\label{D3} A spreadsheet
$S$ is \textit{row-organized,\/} if all ranges to which formulas of $S$
refer are horizontal ones, i.e. rows.

Dually, $S$ is \textit{column-organized\/}, if all ranges to
which formulas of $S$ refer are vertical ones, i.e. columns.

Finally, $S$ is \textit{un-organized\/} if it neither
row-organized nor column-organized.

In all cases, references to the input part of $S$ are exempt
from those limitations.
\end{definition}

\begin{definition}\label{D4}
A spreadsheet $S$ is \textit{row-directed\/} if every formula refers
only to cells and ranges located in rows above that one in which the
formula is located.

Dually, $S$ is \textit{column-directed}, if every formula
refers only to cells and ranges located in columns located to the left
of that one in which that formula is located.

$S$ is \textit{un-directed\/} if it is neither row- nor
column-directed.

Finally, $S$ is \textit{bi-directed\/} if it is both row- and
column-directed.

Again, references to the input part of $S$ are exempt from
those limitations.
\end{definition}

\section{PRAM model}\label{S3}

A PRAM machine $A$ consists of the following components:

\begin{enumerate}
\item Unbounded number of cells of global read-and-write shared
memory, each one equipped with:

\begin{enumerate}
\item Its own serial number, unique, counted from 1 on.

\item Capacity to store one integer, initially set to 0.
\end{enumerate}

\item Unbounded number of cells of global read-only input memory, each
one equipped with:

\begin{enumerate}
\item Its own serial number, unique, counted from 1 on.

\item Capacity to store one integer, initially set to 0.
\end{enumerate}

\item Unbounded number of processors, each one equipped with:

\begin{enumerate}
\item Its own serial number $ s,$ unique, counted from 1 on, and a
constant $ N,$ equal to the total number $p$ of active processors and
a constant $ M$ equal to the number of active cells of shared memory.

\item Three private registers $ i,j,k$ for storing integers, initially set
to 0. 

\item Its own instruction counter, initially set to 1.
\end{enumerate}

\item Program $\pi$, which is a list of consecutively numbered
instructions of the following forms, and where $ x$ is ranging over
$ i,j,k$ and $ u,v$ are ranging over
$ i,j,k,s,N,M,M[i],M[j],M[k],M[s],I[i],I[j],I[k],I[s]$ and integer
constants and $\ell $ ranges over integer constants:

\begin{enumerate}
\item $ x:=u$ 

\item $ x:=u\circ v$, where $\circ $ is among $\{ +,-,*,/\} $ 

\item $ M[x]:=u$ 

\item $ M[x]:=u\circ v$, where $\circ $ is among $\{ +,-,*,/\} $ 

\item \textbf{if} $ u<v$ \textbf{then goto} $\ell $ $ $ 
\end{enumerate}
\end{enumerate}

The number of private registers can be chosen
arbitrarily. We needed some fixed limit, so we have decided to use 3
of them.

$ M[x]$ stands for the shared memory cell whose address is
$ x,$ and $ I[x]$ stands for the input memory cell whose address is $ x.$ 

The input sequence $ v$ of integers is located in the input
memory cells, and $ I[1]$ contains the length of the input sequence,
including itself (so that the empty input sequence is passed to the
machine as 1 in $ I[1]$).

Then a fixed number of its processors (say $p$) is
initialized, and a fixed number of shared memory cells (say $ m$) is
initialized. Typically $p$ and $ m$ are adjusted to the length of $ v.$ 

During the computation each of the processors follows this
program. The state of $A$ at each time $ t$ of its computation on input
$ v$ is defined to be the sequence of $p$ 4-tuples of integers and a
sequence of $ m$ integers: the $ n$-th tuple is the state of the $ n$-th
processor, consisting of: the values of its local variables $ i,j$ and
$ k,$ the value of its instruction counter, while the sequence of
integers represents the content of the shared memory.

Initially (i.e., at time $ t=0$) the values of local
registers are 0, and the instruction counter 1, and the shared memory
values are 0.

A single step of computation of $A$ corresponds to a
parallel, simultaneous change of all $p$ 4-tuples and $ m$ integers
describing the processors and shared memory of $ A.$ 

The effect of a reference to the local registers, in order
to read or write, should be clear, if we say that $ /$ is the integer
division. References to the shared memory and input memory cells are
more complicated. First of all, an attempt to read from a nonexistent
input or shared memory cell (i.e., of address higher than $ I[1]$ or
than $p$, resp.) or to read from cells of numbers smaller than 1 is an
error and the result of this operation is unpredictable: it may cause
the machine to break and stop operating, or to retrieve some
accidental value and continue computation.

An attempt to write to a shared memory cell of zero or
negative number is permitted, but has no effect, and similarly if that
number is higher than $p$. If more than one processor attempts to read
from the same shared or input memory cell, all of them succeed and get
the same value. If more than one processor attempts to write to the
same shared memory cell, all the requests are executed, and the new
value of that memory cell is the one written by the processor with the
lowest serial number; the values written by the remaining processors
get lost. Reading is performed before writing, so the processors which
read from a shared memory cell to which other processors wish to
write, get the ``old'' value. The model of PRAM with this policy of
resolving read and write conflicts is called priority write
\textbf{C}oncurrent \textbf{R}ead \textbf{C}oncurrent \textbf{W}rite
PRAM, abbreviated CRCW PRAM.

Of course, the instruction counter is increased by one in
such cases. The conditional instruction \textbf{if} $ u<v$ \textbf{then
goto} $\ell $ causes the instruction counter of a processor to be
updated depending on the comparison of the values of $ u$ and $ v$ : if
$ u$ is less than $ v$ then it is updated to $\ell $, otherwise it is
incremented by 1. We always assume that $\ell $ is a positive integer
not exceeding the total number of lines in $\pi$-- note that the
value of $\ell $ is hard-coded into the program. If the value of the
instruction counter becomes higher than the number of lines in the
program, the processors halts.

Thus, given a PRAM $A$ as above and its input vector $ v$,
the computation of $ M$ on $ v$ is represented by a finite or infinite
sequence of states of $A$, which may but need not be constant from
some moment on.

The result of computation of $A$ after $ n$ steps is the
content of the shared memory after completing that step.

Another, substantially weaker model of PRAM is EREW, short
for \textbf{E}xclusive \textbf{R}ead \textbf{E}xclusive
\textbf{W}rite, which results from CRCW by forbidding concurrent reads
and writes altogether. Its advantage is that it is closer to real-life
parallel computers than CRCW.

The programming language of our PRAM machines is extremely
simple, but, as it is well-known, equivalent in computing power to
even very rich ones, so indeed each processor separately has a
universal computing power, equivalent to that of a Turing machine.

PRAM is a machine which can easily implement referential data
structures, such as lists, trees, etc., as well as arrays. Therefore
we use them without any further explanation.

\section{Complexity theory}\label{S4}

$p$ is the class of all functions computable by PRAMs with
one processor and unbounded amount of memory, in time polynomial in
the length of the input. In fact, the number of memory cells ever used
in a computation cannot exceed the number of computation steps made.

In this paper we use the $P$-complete problem
\textit{Circuit Value Problem\/} (abbreviated CVP).

An instance of CVP is a sequence of $ n$ Boolean
substitutions (the reason for starting numbering from 2 is purely
technical and explained below):
\[p_{2} :=conn_{2} (inputs_{2})\] 
\[p_{3} :=conn_{3} (inputs_{3})\] 
\[\ldots \] 
\[p_{n} :=conn_{n} (inputs_{n})\] 
The connectives $conn_{i}$ can be \textbf{and}, \textbf{or} and
\textbf{not}. The first two of them have always two inputs,
\textbf{not} has one input. Each of the inputs in $inputs_{k}$ can be
either $true$, or $false$, or a variable $p_{i}$ with $i<k,$
indicating that the value of that variable should be used. For
example, the following is a valid instance of CVP of size $ n=4$ :
\begin{eqnarray*}
p_{2} &:=and(true,false)\\ 
p_{3} &:=or(true,p_{2})\\ 
p_{4} &:=or(false,p_{3})\\ 
p_{5} &:=not(p_{4})
\end{eqnarray*}

One can then calculate the values of all variables in that
instance, proceeding top-down: $p_{2}$ is $ false$, $p_{3}$ is
$true$, $p_{4}$ is $true$ and $p_{5}$ is $false.$ 

The CVP problem is that, given an instance of CVP, to decide if the
last variable is $true$. That problem is known to be $P$-complete.

We encode CVP in various spreadsheets in order to demonstrate that
they are unlikely to be efficiently parallelizable. For convenience,
when we do so we use 0 in place of $false$, 1 in place of $true$,
for variables we use their numbers as names (we have started numbering
from 2, so this does not lead to confusing truth values with
variables), and we drop all conventional symbols like $:=$,
parentheses and commas, so that the above example is encoded by the
spreadsheet shown in Figure \ref{F3}.

\begin{figure}
\includegraphics[width=1.5in]{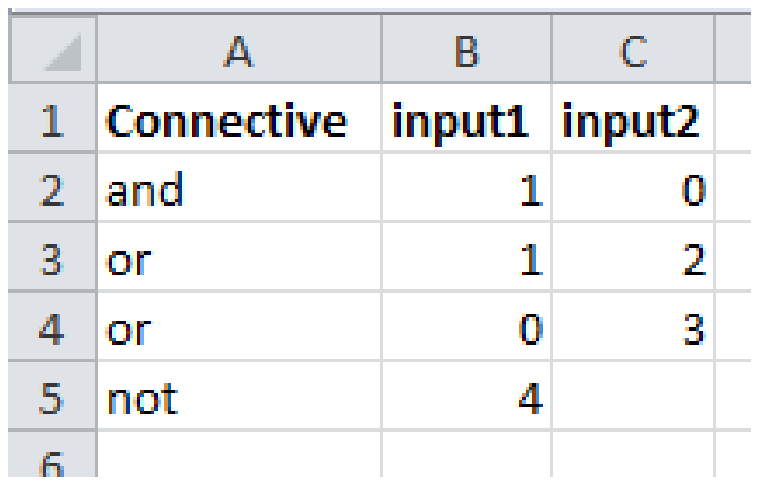}
\caption{CVP instance in a spreadsheet.}
\label{F3}
\end{figure}

$NC$ is the class of functions computable by PRAMs with
polynomially many processors, using polynomially many shared memory
cells and in time bounded by $\log ^{k} n$ for some constant $ k$, for
input data of size $ n$.

We often use the ``big Oh'' notation for the asymptotic growth of
functions. In this notation, $P$ is the class of functions computable
by a sequential (i.e., single-processor) program in time $n^{O(1)}$,
while $NC$ is the class of functions computable in time $\log^{O(1)}
n$ on a parallel computer using $n^{{O(1)}} $ processors.

It is a well-known fact that $ NC\subseteq P,$ and it is
widely believed, but remains unproven, that indeed $ NC\ne P.$ 

This hypothesis is often rephrased as a statement, that there are
problems decidable in polynomial time, which are inherently
sequential, i.e., cannot be solved very fast by parallel machines. $
NC$ is understood in this statement as ``the'' class of problems which
can be solved fast in parallel.

If $ NC\ne P$, then in particular $P$-complete problems cannot be
solved by PRAMs in time $\log^{O(1)} n$.

Therefore CVP cannot be solved on a PRAM with polynomially
many processors and polynomially many cells of shared memory in time
$O(\log^{O(1)} n)$, where $n$ is the size of input
and $ k$ is a constant, unless $P=NC.$ 

In this paper we do not derive hypothetical lower bounds on
the complexity of evaluating various classes of spreadsheets. Instead,
we estimate the size of CVP instances, which are computable in the
spreadsheets in question. Each time it is easy to translate the size
of CVP instances we produce into the potential lower bounds.

\section{Un-directed spreadsheets I: Complexity}\label{S5}

It is obvious, that if the cells of a small initial spreadsheet $S$
are filled to create $c$ columns and $r$ rows of formulas, then the
resulting spreadsheet can be computed in time polynomial in $cr$,
given the initial $S$, the dimensions $c$,$r$ and the input data of
$S$. More precisely, spreadsheet operators typically can be
implemented by algorithms of time complexity linear in the size of
their input ranges, with a few exceptions like \texttt{FREQUENCY}, which
require complexity of order $O(n\log n)$. The largest possible range
in a spreadsheet with $cr$ cells does not exceed $cr$ cells, so
computing cells can be done in time $c^{2} r^{2}\log cr$. Determining
the evaluation order requires topological sorting, which can be done
in time linear in the number of inter-cell references, which does not
exceed $ c^{2} r^{2} $. Therefore in total we get the maximal possible
computation time $O(c^{2}r^{2} \log cr)$.

In case of row-organized spreadsheets or column-organized
spreadsheets we get the orders of $ cr^{2} \log cr$ and $ c^{2} r\log
cr$, respectively, due to restricted sizes of ranges.

The following theorem is neither surprising nor difficult to
demonstrate. It states in particular, that CVP is computable by
spreadsheets.

\begin{theorem}\label{T1}
There exists an un-directed, un-organized spreadsheet $S_{4}$, such
that when it is extended to dimensions of either $n$ rows and $6$
columns or $6$ rows and $n$ columns, it computes the solution to the
CVP problem of size $n$, given its description as input.
\end{theorem}
\begin{proof} A fully commented spreadsheet is provided as electronic
  appendix S4. 
\end{proof}

The same behavior is demonstrated by a slightly more complicated
spreadsheet $S_{5}$, which extended to yield $4r$ rows and $2c$
columns, computes the answer to a CVP instance of size $cr.$ So its
full area is devoted to solving the instance of CVP. This spreadsheet
is provided as electronic appendix S5.

The difference between $S_{4}$ and $S_{5}$ is that the former consists
of two functionally separate fragments, which can be independently
converted into row-oriented and column-oriented structure, as it is
demonstrated in electronic appendix S4. In $S_{5}$ it is impossible,
because it uses two-dimensional ranges in a crucial way.

\section{Directed spreadsheets I: Algorithm for a restricted function set}\label{S6}

At first, we restrict spreadsheets significantly, in order
to formulate and prove our results in a simplified scenario.

First, we restrict spreadsheets to use only integer numbers.

Next, we restrict spreadsheets to just a few functions, here
given in \textit{Excel\/} syntax.

All of them are also found (perhaps in a slightly different
syntactical form) in the competing spreadsheet systems, which
therefore are also covered by the theorems we shall prove below.

\subsection{Functions in restricted spreadsheets}

\subsubsection{Arithmetical functions}

We use standard arithmetical functions: $ +$, $\cdot $, $-$ 
and $ /$.

\subsubsection{Comparison functions}

\texttt{value1=value2}

\texttt{value1<value2}

\texttt{value1>value2}

\texttt{value1<=value2}

\texttt{value1>=value2}

\texttt{value1<>value2}

\texttt{value1} and \texttt{value2} can be numbers, formulas or cell
references to numbers. The result is \texttt{TRUE} if the arguments
are in the specified relation, and \texttt{FALSE} otherwise.

\subsubsection{Logical functions}

\texttt{IF(test,value1,value2)}

If \texttt{test} is, refers or evaluates to \texttt{TRUE}, the
function returns \texttt{value1}, if \texttt{test} is, refers or
evaluates to \texttt{FALSE} it returns \texttt{value2}. In all other
cases the result is a \texttt{\#VALUE!} error.

\texttt{IFERROR(value1,value2)}

This function returns \texttt{value2} if \texttt{value1} is, refers to or
evaluates to any error value, and \texttt{value1} otherwise.

\texttt{CHOOSE(index-num,value1,value2,\dots)}

\texttt{CHOOSE} is a kind of generalization of \texttt{IF}, because in one
formula it allows the choice among up to 29 possible values to be
returned. \texttt{index-num} specifies which value argument is
selected. \texttt{index-num} must be a number between 1 and 29, or a formula or
reference to a cell containing a number between 1 and 29.

If \texttt{index-num} is $i$, \texttt{CHOOSE} returns \texttt{value$i$}.

\texttt{AND(value1,value2,\dots)} and \texttt{OR(value1,value2,\dots)}

compute the logical conjunction and disjunction of their
arguments, respectively.

\subsubsection{Address functions}

We use two address functions: \texttt{ROW()} and \texttt{COLUMN()}, which
return the number of row (column, resp.) od the cell in which they are
located. In case they are given an argument, a reference to a single
cell, they return the row (column, resp.) in which that reference is
located.

\subsubsection{Aggregating functions}

\texttt{MATCH(lookup-value,lookup-array,match-type)}

\texttt{lookup-value} is the value to be found in a range specified
by \texttt{lookup-array}. \texttt{lookup-value} can be a number, a formula or a cell
reference to a number.

\texttt{match-type} in the spreadsheets we create is typically 0. It
specifies how the spreadsheet matches \texttt{lookup-value} with values in
\texttt{lookup-array}: \texttt{MATCH} finds the first value that is exactly equal to
\texttt{lookup-value} and returns its relative position in the \texttt{lookup-array}.

If \texttt{match-type} is 1, then values in \texttt{lookup-array} are assumed
to be sorted into an increasing order, and \texttt{MATCH} finds the first value
that is larger or equal to \texttt{lookup-value} and returns its relative
position in the \texttt{lookup-array}. If match-type is -1, then values in
\texttt{lookup-array} are assumed to be sorted into a decreasing order, and
\texttt{MATCH} finds the first value that is smaller or equal to \texttt{lookup-value}
and returns its relative position in the \texttt{lookup-array}.

In all three cases, if no value is found which satisfies the
criteria, the result is an error \texttt{\#N/A!}. In case of \texttt{match-type} equal
$\pm 1$ \texttt{lookup-value} is larger (smaller, resp.) than all values in the
sorted \texttt{lookup-array}, the result is the number of the last non-empty
cell in \texttt{lookup-array}.

\texttt{INDEX(array,row-num,col-num)}

\texttt{array} is a range of cells, \texttt{row-num} and \texttt{col-num} can be
numbers, formulas yielding numbers or cell references to numbers. The
result of the function is a value whose relative position in \texttt{array} is
given by \texttt{row-num} and \texttt{col-num}. We use this full syntax exactly twice,
otherwise array is a single column or row and one of the arguments is
missing -- its default value is 1.

\subsection{Simulating spreadsheets by PRAMs}

In this section, we are going to formulate and prove a
theorem about evaluating spreadsheets by PRAM machines. In our model
spreadsheets can be of unbounded size, so we can use asymptotic
notation to describe the resources needed by a PRAM to execute a
spreadsheet of a given size. The theorem below is formulated for
row-organized row-directed spreadsheets, but of course its dual form
for column-organized column-directed spreadsheets holds, too.

\begin{theorem}\label{T2}
For any row-organized
row-directed restricted spreadsheet $S$ there exists a
program $\pi$ for EREW PRAM, such that if that spreadsheet filled
to make $c$ columns and $r$ rows, the values of all its cells can be
computed by $\pi$ run for $O(r\log cr)$ time on $c$ processors and
$cr$ cells of memory, given the initial $S$, $c$, $r$ and
the input data of $S$.

If $S$ is additionally row-organized, then the
values of its cells in the last row can be computed by $\pi$ run
for $O(r\log c)$ steps on a PRAM with $c$ processors and $c$ cells
of memory.
\end{theorem}
\begin{proof} 
For each column of the spreadsheet we designate one
processor, which will be responsible for it. Let the number of that
processor will be equal to the number of the column.

Assume that first from the initial constant size range of
$S$, rows are extended to the right to create $c$ columns, one by one
starting from the top row. Then a number of so created rows is filled
downwards to create $r$ rows.

The idea is to organize the computation of PRAM simulating
the above procedure in rounds, where each round corresponds to
creating one new row.

In order to evaluate formulas creating the row, PRAM must
know what to evaluate. The codes for evaluating particular formulas
are hard-coded into the program $\pi$ it executes. Observe
that the initial $S$ has a constant number of rows and
formulas, and they repeat in a circular fashion in the large
spreadsheet created by filling. Hence the program is divided into a
constant number of branches, one for each type of a column, which are
executed by the processors responsible for such columns. The portion
of code for each type of column is composed of a main loop, the body
of which determines how to evaluate the constant number of formulas
repeating in the column. Evaluating of formulas sometimes requires
coordinated actions of all processors. The code to be evaluated in
such cases is also provided in the body of the main loop.

For each round we assume that certain auxiliary data
structures are available, which enable execution of the round
sufficiently fast by PRAM. During each round, first the new values are
computed, and then these structures are updated, so that they include
the cells in the newly created row, as well.

Separately, we must explain how the auxiliary data
structures are initialized before the first round, when they should
contain information about input data.

\subsubsection{Auxiliary data structures}

We assume that during each round all the previously created
columns and rows are stored in two copies. Each row is stored in two
copies:

\begin{enumerate}
\item  in the original form,

\item  as a sorted array, where we sort and store two-element
records consisting of the value from the original row and its original
address.
\end{enumerate}
Each column is stored in two copies:

\begin{enumerate}
\item in the original form,
\item  as a balanced binary search tree (say: red-black tree), in which we
store two-element records consisting of the value from the original
row and its original address. 
\end{enumerate}

\subsubsection{Initialization of auxiliary data structures.}

Initially there is a constant number of rows and/or columns
without formulas, containing input data. It is easy to see that we can
relocate columns of input data into rows.

In order to create their sorted forms, we perform a logarithmic time
parallel sorting algorithm using $c$ processors and linear amount of
additional memory, like the one described in \cite{Rytter}[Section
  5.2]. This initialization takes $O(\log c)$ time.

Initial columns are of constant height, so the red-black
tree of each of them is created by one processor, responsible for that
column, in constant time.

\subsubsection{Execution of a round. Computing formulas}

Henceforth, PRAM must first evaluate formulas. The formulas to be
computed refer only to data above themselves, so all of them can be
evaluated independently in parallel. For each of the cells, it is done
by a single processor, responsible for the column of that cell. The
values of all functions except \texttt{MATCH} and \texttt{INDEX} can
be obviously evaluated in constant number of steps. Note that
\texttt{COLUMN} function can be evaluated, because each processor
knows its serial number, equal to the column number. \texttt{ROW} on
the other hand is evaluated by keeping a constant record of the
advancing computation time.

If \texttt{MATCH} looks up a row, a sorted version of that row is in
the auxiliary data structures, in which the processor can find the
suitable value (and its accompanying address, which ought to be
returned) using binary search in time $O(\log c).$ 

If MATCH looks up a column, a red-black tree version of that
column is in the auxiliary data structures, in which the processor can
find the suitable value (and its accompanying address, which should be
returned) in time $O(\log r).$ 

In total, computing the new values takes $O(\log c+\log
r)=O(\log cr)$ time.

\subsubsection{Execution of a round. Updating auxiliary data structures.}

The sorted copy of the newly created row is computed in
$O(\log c)$ time using logarithmic time linear memory sorting
employing all $c$ processors. Then all processors in parallel insert
the new values from their columns into the corresponding red-black
trees, in time $O(\log r)$.

In total, updating the necessary data structures takes
$O(\log cr)$ time.

\subsubsection{Cost of the algorithm}

First, initialization of auxiliary data structures takes
$O(\log cr)$ time. Then the simulation of the PRAM computation
performs $r$ rounds, each of them takes $O(\log cr)$ time, so the
total time is $O(r\log cr).$ 

The memory used is constant times $cr$.

For the second claim, in a row-organized spreadsheet there
is no need to access columns, so we do not need to maintain red-black
trees. Thus in this modified version each round can be completed in
$O(\log c)$ time and the total running time is $O(r\log c).$ 

Moreover, if initially $S$ had $k$ rows, then every formula
copied from $S$ either accesses one of the initial $ k$ rows, or row at
most $k$ levels above itself. Consequently we need to store at most
$2k$ rows simultaneously, and therefore the total amount of necessary
memory is $O(c).$ 
\end{proof}

\section{Directed spreadsheets II: Algorithm for a rich function set}\label{S7}

The most significant differences between real-life
spreadsheets and the simplified ones, which we have discussed so far,
are:

\begin{enumerate}
\item More data types than just integers, most notably floating point
numbers and text strings.

\item More functions are permitted.
\end{enumerate}

\subsection{More functions}

Generally, functions provided in spreadsheet systems can be
divided into two types basic types: scalar ones, which have only
several single cell arguments (including functions without any
arguments), and aggregating ones, which have at least one argument,
which is a range. There are functions which can be used both as scalar
and aggregating ones. Therefore we apply this distinction to formulas
in spreadsheets, where each function occurrence in a formula is,
depending on the particular syntax used in it, either scalar or
aggregating.

\subsubsection{Scalar functions}

Scalar functions can be easily incorporated into our
discussion. If a scalar function is added simultaneously to the syntax
of PRAM programs and to the spreadsheets, the two systems remain in
the same relation as indicated in our theorems. The proofs can easily
be adapted.

\subsubsection{Aggregating functions}

In order to avoid considering hundreds of functions found in real life
spreadsheets, the approach we take here is to consider \textit{array
  formulas\/}, which can equivalently express great majority of
aggregating functions present in spreadsheets, and many functions
which cannot be expressed by other means. After introducing them, we
formulate a general condition on array formulas, which guarantees that
spreadsheets which use them can be evaluated by a PRAM within the same
time bounds as \texttt{MATCH} and \texttt{INDEX} functions, the only
aggregating functions we permitted so far. In particular, this means
that Theorem 2 remains true for spreadsheets which contain array
formulas which meet this condition.

The array formulas we use in this paper are constructed as
follows:

First the \textit{inner expression\/} is formed from scalar
functions, but whose arguments are rows or columns (of identical
sizes) instead of single cells. Single cells or constants appearing in
the formula are interpreted as rows or columns of identical
values. Then the result of this expression is again a row or column,
resulting by computing the functions at each position
independently. In particular, an array formula cannot contain rows and
columns simultaneously.

Finally, the inner expression is an argument of an aggregating
function, which takes an arbitrary number of elements and returns a
single number, such as, e.g., \texttt{SUM}, \texttt{MAX} or
\texttt{AVG}. This function is applied to the row or column of results
of the expression and yields a single number, which is the value of
the array formula.

Array formulas are entered by pressing Ctrl+Shift+RETURN and
are marked by a pair of braces \texttt{\{\}} around, which are introduced by
the system (and not typed in by the user). This identifies them for
the spreadsheet system and invokes their special evaluation algorithm.

\begin{example}\label{E1} Let us consider the following the
formula and its context shown in Figure 4.

\begin{figure}
\includegraphics[width=3in]{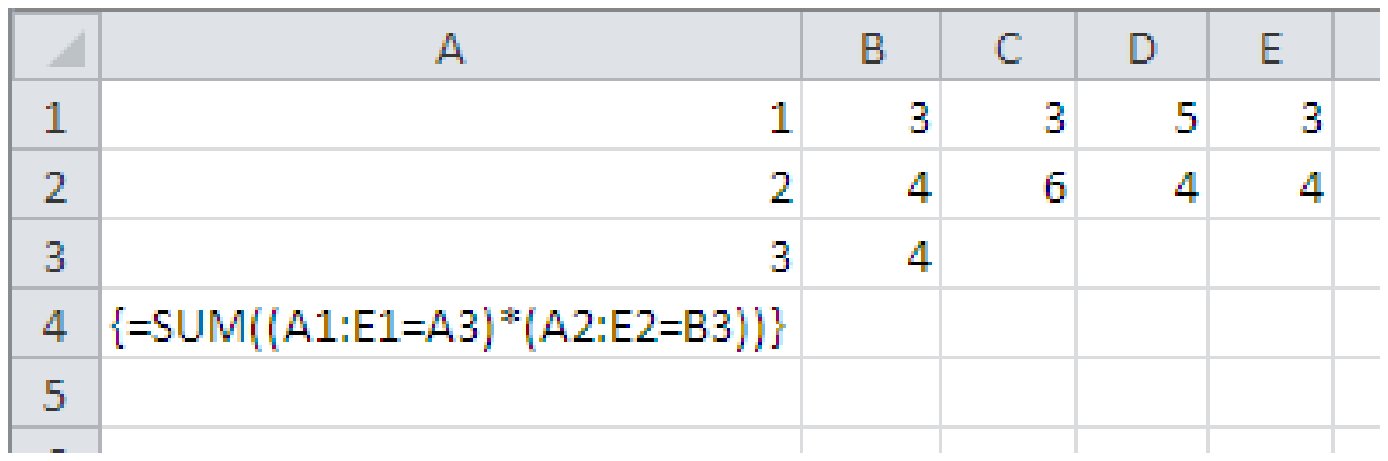}
\caption{Array formula and its context}
\label{F4}
\end{figure}

The outer aggregating function is \texttt{SUM}, and its inner
expression is \texttt{(A1:E1=A3)*(A2:E2=B3)}. In this situation
\texttt{A3} and \texttt{B3} are treated as rows of five 3s and five
4s, respectively. This way equality can be used to compare a range and
a single cell, and returns, for position in the rows,\texttt{ TRUE} if
its arguments are equal and FALSE otherwise. For multiplication,
\texttt{TRUE} is treated as 1 and \texttt{FALSE} as 0. Carrying out
these calculations at each position independently, we get a row of
values 0, 1, 0, 0, 1. Now \texttt{SUM} applied to it yields 2 and this
is the result of the array formula. Thus the final result is the count
of positions, at which rows 1 and 2 contain the pair $ (3,4)$ (i.e.,
the same as in cells \texttt{A3} and \texttt{B3}). Consequently, we
get the same function as computed by
\texttt{COUNTIFS(A1:E1,A3,A2:E2,B3)}, aggregating function present in
the newest editions of \textit{Excel\/}.
\end{example}

\begin{example}\label{E2} Another example consists of five formulas
is shown in Figure 5.

\begin{figure}
\includegraphics[width=3in]{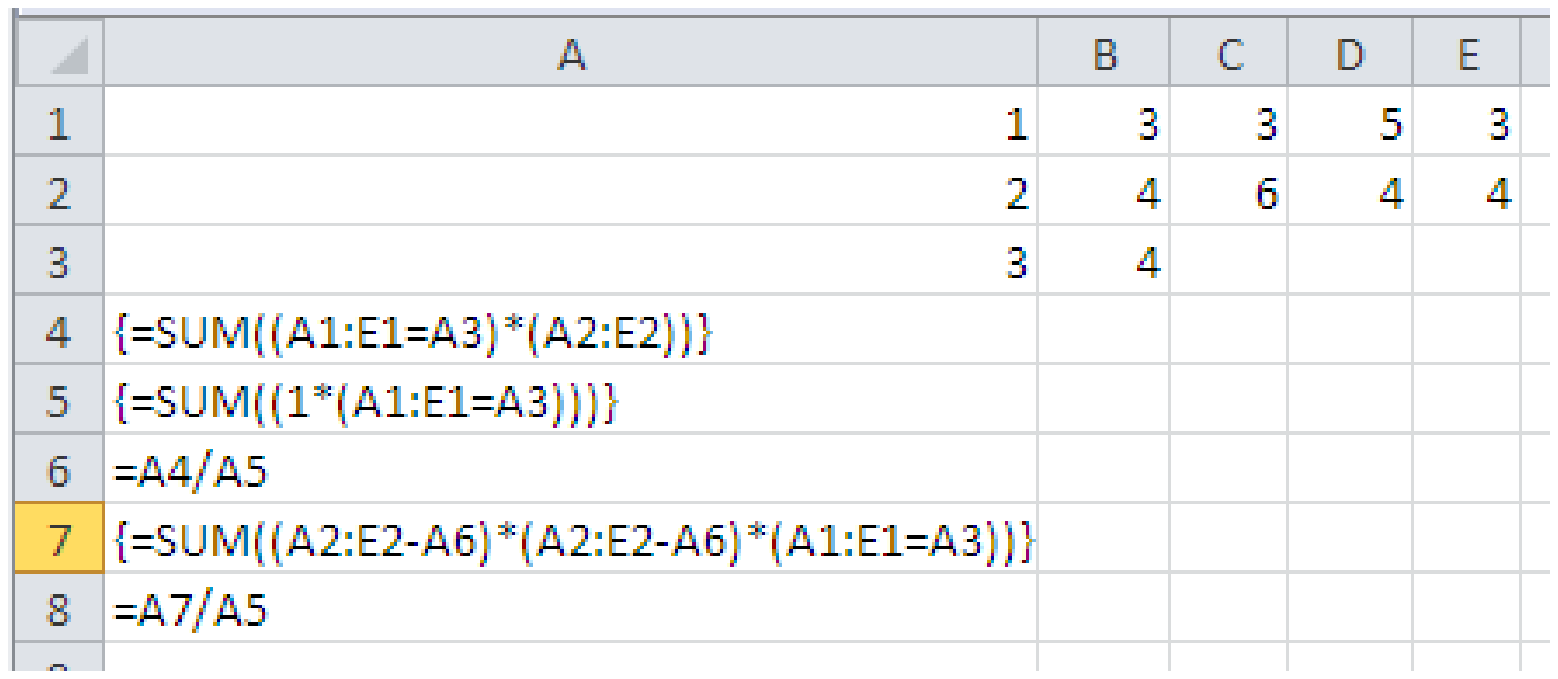}
\caption{Five array formulas and their
context}
\label{F5}
\end{figure}

The formulas compute the variance of the subset of those
numbers in row 2, which are accompanied in row 1 by a value identical
to \texttt{A3}. Indeed, the first formula (located in \texttt{A4}) computes the sum of
those values, the next formula (located in \texttt{A5}) calculates the number
of those values, the one in \texttt{A6} calculates the mean, and finally the
next two, found in \texttt{A7} and \texttt{A8}, calculate the variance according to the
formula ${\mathbb V}(X)={\mathbb E}((X-{\mathbb
E}(X))^{2}).$ 

The actual numerical result is not so important. The formula
in A6 computes the equivalent of the \texttt{AVERGAGEIFS} function found in the
newest versions of \textit{Excel}, while \texttt{A8} computes the equivalent of
the \texttt{DEVPIFS} function, that hypothetically would compute the variance
of the whole population of elements satisfying multiple criteria -
however, this function is not present in \textit{Excel}.
\end{example}

We have given an example of a spreadsheet function, which
can be equivalently expressed as an array formula, and example of an
array formula whose functionality cannot be expressed by a single
standard spreadsheet function.

\begin{theorem}\label{T3}
For any row-organized row-directed spreadsheet $S$ which, in addition
to functions permitted in restricted spreadsheets, uses array formulas
which

\begin{enumerate}
\item have \texttt{SUM} as the outer aggregating
function, 

\item whose inner expressions use only multiplication, addition, subtraction
and comparison operators\texttt{=}, \texttt{<=}, \texttt{<},
\texttt{>=}, \texttt{>}, \texttt{<>}, applied directly to compare
argument ranges and cells,

\item moreover, have at most one comparison operator
other than \texttt{=} applied to compare a range to a single cell,
\end{enumerate}

there exists a program $\pi$ for EREW PRAM, such
that if that spreadsheet filled to make $c$ columns and $r$ rows,
the values of all its cells can be computed by $\pi$ run for
$O(r\log cr)$ time on $c$ processors and $cr$ cells of memory, given
the initial $S$, $c$, $r$ and the input data of $S$.

If $S$ is additionally row-organized, then
the values of its cells in the last row can be computed by $\pi$ 
run for $O(r\log c)$ time on $c$ processors and $c$ cells of
memory. 
\end{theorem}
\begin{proof} It suffices to demonstrate how the
evaluation of an array formula of the above described form can be
incorporated into the algorithm described in the proof of Theorem 2.

The inner expression is built from ranges, cells and
comparisons using multiplication, addition and subtraction. In
particular, the inner expression can be thus equivalently expressed as
a sum of products of ranges, cells and comparisons. Since the outer
aggregating function is \texttt{SUM}, it suffices to explain how to evaluate an
array formula whose inner expression is a single product of ranges,
cells and comparisons. Single cells as factors in the product can be
moved outside of the outer \texttt{SUM}, hence finally we have an array formula
of the form: outer aggregating function \texttt{SUM} and inner expression is a
product of expressions of the following types:

\begin{enumerate}
\item  ranges \texttt{U},

\item  comparisons between ranges \texttt{U$r$W}, where $r$ is any of
\texttt{=}, \texttt{<=}, \texttt{>=}, \texttt{<},  \texttt{>}, \texttt{<>},

\item  equality comparisons between ranges and single cells \texttt{V=C},

\item  at most one inequality comparison between a range and a
single cell \texttt{X$r$D}, where $r$ is any of \texttt{<=}, \texttt{>=}, \texttt{<}, \texttt{>}.
\end{enumerate}

\subsubsection{Auxiliary data structures}

If the ranges in the array formula are rows, the auxiliary data
structure is prepared as follows: for every subset of already existing
rows, which potentially may play together the roles of \texttt{U} and
\texttt{W} ranges above, the values found in the ranges and the
results of comparisons \texttt{U$r$W} are multiplied in each column
separately. Then the records formed from the resulting values together
with the values found in the rows which play the role of ranges
\texttt{V} are sorted according to the values in the \texttt{V} rows
and the \texttt{X} row. Which of the rows \texttt{V} are more or less
significant in the sorting is unimportant, but the single \texttt{X}
row (if present) is the least significant one in the sorting.

Finally, for the so created array prefix sums of results of
multiplications are computed and stored as an array (i.e., the first
value in the prefix sums array is the first value from the original
sorted array, the second value in the prefix sums array is the sum of
the first and second value from the original sorted array, the third
value in the prefix sums array is the sum of the first, second and
third value from the original sorted array, and so on).

If the ranges in the array formula are columns, the auxiliary data
structure is prepared as follows: for every subset of already existing
columns, which potentially may play together the roles of \texttt{U}
and \texttt{W} ranges above, the values found in the ranges and the
results of comparisons \texttt{U$r$W} are multiplied in each row
separately. Then the records formed from the resulting values together
with the values found in the columns which play the role of ranges
\texttt{V} are stored in an red-black tree according to the values in
the \texttt{V} rows and the \texttt{X} row. Again the single
\texttt{X} row (if present) is the least significant one in the
sorting. Apart from that, in the red-black tree each node stores
additionally the sum of all results of multiplications in the subtree
rooted at this node.

\subsubsection{Initialization of auxiliary data structures.}

Initially there are only constantly many rows without
formulas, consisting of input data.

In order to create their sorted forms, we perform a logarithmic time
parallel sorting algorithm using $c$ processors and linear additional
memory. Then prefix sums are computed according to the algorithm
described in \cite{Cormen}[Section 30.1.2]. In total, the
initialization takes $O(\log c)$ time.

Initial columns are of constant height, so the red-black tree of each
of them is created by one processor, responsible for that column, in
constant time. Computing sums of values from column \texttt{X} is also
immediate. This initialization takes constant time.

\subsubsection{Execution of a round. Computing formulas}

PRAM must first evaluate array formulas in the newly created row. The
formulas to be computed refer only to data above themselves, so all of
them can be evaluated independently in parallel. For each of the
cells, it is done by a single processor, responsible for the column of
that cell.

In case of row ranges, when the values of single cells \texttt{C} are
already known, binary search is performed on the ordered array to find
the first and the last element whose \texttt{V} values agree with the
\texttt{C} values. Between them there are records with identical
values in V rows, ordered by their values in row \texttt{X}. Using
further binary search, we find the point indicated by the value from
cell \texttt{D}. The values in the interval between this point and one
of the previously identified elements should be summed up. In order to
do so, we use the array with prefix sums. The required sum is the
difference between prefix sum corresponding the rightmost of them
minus prefix sum corresponding to the leftmost of them. This procedure
takes $O(\log c)$ time.

In case of column ranges, when the values of single cells \texttt{C}
are already known, the procedure is similar to the one used for
rows. The key observation is that if the processor searches for an
element in the tree, it can sum the sums recorded in the left sons of
nodes at which it descends towards the right son. When the required
element is found, the sum of those sums is precisely the prefix sum
corresponding to that element. With this remark, the previously used
computation can be easily simulated, and takes $O(\log r)$ time.

In total, computing the new values takes $O(\log c+\log
r)=O(\log cr)$ time.

\subsubsection{Execution of a round. Updating auxiliary data structures.}

The sorting of the auxiliary rows and creating the array with prefix
sums are performed in $O(\log c)$ time using the algorithms we have
used in initialization, using all $c$ processors.

Updating red-black trees by inserting the new values from their
columns and updating the sums is done in $O(\log r)$ time.

In total, updating the necessary data structures takes
$O(\log cr)$ time.

\subsubsection{Cost of the algorithm}

First, initialization of auxiliary data structures takes $O(\log cr)$
time. Then PRAM performs $r$ rounds, each of them takes $O(\log cr)$
time, so the total time is $O(r\log cr).$

The memory used is constant times $cr$.
\end{proof}

Generally, the following standard, important functions found in
spreadsheet systems can be equivalently expressed by (one or a few)
array formulas which obey the criteria of Theorem \ref{T2}:
\texttt{SUM}, \texttt{SUMIF}, \texttt{COUNT}, \texttt{COUNTIF},
\texttt{AVG}, \texttt{AVGIF}, \texttt{STDEV}, \texttt{SUMX2PY2},
\texttt{SUMX2MY2}, \texttt{SUMPRODUCT}.

Two more functions, introduced for the first time in \textit{Excel\/}
2007: \texttt{SUMIFS} and \texttt{COUNTIFS} are covered in a limited
form: if at most one of the criteria in the function is an inequality,
then it can be expressed by an array formula satisfying the
assumptions of Theorem \ref{T3}.

\subsection{Floating point numbers and strings}

Floating point numbers can be added to the PRAM model by allowing the
register and memory cells to store them. All theorems formulated for
relations between classical PRAMs and spreadsheets with integers hold
analogously when formulated for PRAMs and spreadsheets with floating
point numbers, and the proofs remain the same.

Text strings do not pose any significant problem. Indeed, there is no
aggregating function for strings in \textit{Excel}. Hence, the only
functions which use them are scalar ones, and they can be added very
much the same way as scalar functions for integers.

\section{Directed spreadsheets III: What can they compute?}\label{S8}

At this point, we have demonstrated that directed spreadsheets can be
evaluated by quite efficient parallel algorithms.

Two questions arise naturally:

Is the estimate from the previous section optimal, or is there a
chance for further improvement?

What is the set of functions, which can be implemented in directed
oriented spreadsheets?

In order to answer both questions simultaneously, we are going to
demonstrate now that there exists a spreadsheet program, with the
following property: any given CRCW PRAM $A$ can be encoded in a can be
simulated by a row-organized row-oriented spreadsheet.

\begin{theorem}\label{T4}
There exists a row-organized row-directed spreadsheet using restricted
function set, which is able to simulate any PRAM $A$ for which $ p=m$,
so that columns rows correspond to processors of $A$ and rows
correspond to computation time.

Precisely speaking, there exists a single spreadsheet $S_{1}$
consisting of 19 cells (\texttt{A2} to \texttt{A20}) with formulas, one
row for input and a separate input area for input interpreted as a
program, such that for every CRCW PRAM $A$ with program $\pi$, $p$
processors and $p$ cells of shared memory, and for every input vector
$ v$ for $A$, if one

\begin{enumerate}
\item pastes the encoding of $\pi$ into the program area of $S_{1}$

\item marks and fills the initial range \texttt{A2:A20} to the right
  creating $p$ columns (corresponding to the processors and shared
  memory cells of $A$)

\item selects the rows from 11 to 20 of these $p$ columns and fills
  downward so that the bottom row of the new range is $ 10t+10$,

\item pastes into $ S_{1}$ the input vector $v$ of $A$ in the first
  row,
\end{enumerate}

then the cells of the bottom 10 rows compute the state of $A$ after $
t$ steps of computation on $ v$.

This means, that the spreadsheet created from $S_{1}$ in steps 1, 2
and 3 performs the first $t$ steps of the computation of $A$ on every
input, i.e., it is a simulator of $A.$
\end{theorem}

\begin{proof} Electronic appendix S1 is the implementation of
PRAM in a spreadsheet with explanation of the formulas used for that
purpose.

Below we highlight the main issues of the construction of $S_{1}$.

Because $p=m$ in $A$, we may assign a processor to each shared memory
cell and make it responsible for its operation. So in every step of
computation it will perform its own action, and then take care of the
actions of the memory cell.

Conceptually, the main difficulty in the construction which is
presented below is that PRAM is a ``read and write'' machine, i.e.,
every processor can both read data from and write data to any shared
memory cell, while spreadsheet is a computation model in which cells
(which can be thought of as simple processors) can read data only from
other locations, and are allowed to write only to the shared memory
cell they are associated with, and moreover only once. This means that
in the course of simulation we will have to simulate writing by other
means. The idea is that any processor willing to write its contents to
some shared memory cell, has to announce this. Then all processors use
function \texttt{MATCH} to search for the announces of writes to their
shared memory cells, and if there is one, fetch using \texttt{INDEX}
the leftmost one, according to the priority write CRCW conflict
resolution policy.

Because processors located in spreadsheet cells are executed only
once, it is necessary to simulate a PRAM by a spreadsheet which has a
separate cell for every processor at every time instant, so the number
of necessary cells is linear in the number of processors of the PRAM
multiplied by the running time of the PRAM.
\end{proof}

At the time of this writing, \textit{Excel\/} 2013 (TM) allows the
maximal size of a single worksheet to be 1,048,576 rows by 16,384
columns, hence, according to our construction, S1 can (theoretically)
simulate a PRAM with about 16 thousands processors and cells of shared
memory, and running for slightly more than 100 thousands parallel
steps. In practice it would be extremely slow if used for simulating
PRAM computations of the maximal size, and might require very large
amounts of RAM to be installed on the computer to fit in memory. We
also provide another simulator $S_{6}$, permitting structural
programming with \textbf{while}-\textbf{endwhile} and
\textbf{if}-\textbf{endif} rather than \textbf{goto} jump
instructions. It is available as electronic appendix S6.

This construction provides answer for our questions:

\begin{enumerate}
\item There is not much room for improvement over Theorems 2 and
  3. Indeed, they offer $O(r\log cr)$ and $O(r\log c)$ algorithms for
  PRAM with $c$ processors and $O(cr)$ and $O(c)$ cells of memory,
  depending on whether the spreadsheet is only row-directed or
  additionally also row-organized. On the other hand, every PRAM
  computation taking time $ t$, using $c$ processors and $c$ cells of
  memory can be simulated in a spreadsheet with $c$ columns and about
  $ 10t$ rows.

\item The set of computations expressible in row-directed spreadsheets
  is indeed very rich, and it includes a natural parallel complexity
  class.
\end{enumerate}

It is worth noting, that accordinng to the results already proven, we
have the following.

\begin{corollary}\label{C1}
Spreadsheets $S_1$ and $S_6$ are universal row-oriented row-organized
spreadsheets.
\end{corollary}
\begin{proof}
Given a row-oriented row-organized spreadsheet $S$, one can derive an EREW
PRAM program $\pi$, computing the same function as $S$. This $\pi$
can be encoded and provided as (a part of) input to either $S_1$ or
$S_6$, which can execute it. 

If the initial $S$ has $c$ columns and $r$ rows, then $\pi$ should be
run on a PRAM with $O(c)$ processors and memory cells and $O(r\log c)$
time. Then, in order to simulate it, $S_1$ and $S_6$ need $O(c)$
columns and $O(r\log c)$ rows. 

Thus the overhead of simulating a spreadsheet by the universal one is
logarithmic, typical for other universal devices.
\end{proof}

The interest in the corollary is that $S_1$ and $S_6$ use only
functions from the restricted function set of Section \ref{S6}.
Therefore this result indicates that the restricted function set forms
a kind of core of the spreadsheet language of formulas, at least for
the row-oriented row-organized ones. 

Thus, in an attempt to create a theoretical model of spreadsheets,
this restricted function set is a candidate to serve as the set of
basic operations, from which the remaining one can be defined. It
would be very much like the relational algebra and its role in the
theoretical formalization of relational databases.  However, one
should remember that our proof concerns only row-oriented
row-organized spreadsheets, created from a small initial spreadsheet
by filling. 

\section{Directed spreadsheets: A parallel programming
  language}\label{S9} 

Because of the tight relationship between row-organized row-directed
spreadsheets and PRAM computations, we postulate that spreadsheets can
become a programming language for specifying scientific parallel
computations on very large datasets. An important factor is that the
computation time of the PRAM and its memory consumption are directly
related to the dimensions of the spreadsheet, and thus can be
reasonably well controlled by the programmer. Moreover, parallel
evaluation algorithms for such spreadsheets are designed for EREW
PRAM, closer to practical parallel computing facilities than CRCW
PRAM.

From the end-user's point of view, specifying massive parallel
computations by designing spreadsheets does not require prior
programming experience and is likely to be accessible for scientists
and industry. It can be offered according to the increasingly popular
``Software as Service'' paradigm.

Below we describe how such system might work.

\begin{enumerate}
\item The user prepares a small initial spreadsheet application.

\begin{enumerate}
\item The user designs the formulas.

\item The user tests the formulas carefully on a small portion of data
  using \textit{Excel\/} or some other existing spreadsheet software on
  a single machine, filling the spreadsheet with the prepared
  formulas.

\item The user estimates the amount of data to be processed.
\end{enumerate}

\item The user submits the spreadsheet and the complete set of data to
the parallel computation facility, indicating:

\begin{enumerate}
\item the location where the input data should be pasted into the
  spreadsheet

\item the number of rows and columns for which the initial spreadsheet
formulas should be filled.
\end{enumerate}

\item The facility estimates the computation cost and asks the user
for permission to execute it.

\item The user accepts the cost or aborts the procedure.

\item The facility executes the task.

\item The facility sends the computation results back to the user.

\item The facility asks the user if he/she is interested in extending
the computation on the same dataset.

\item The user makes the decision.

\item If the user considers continuation of the computations, the
facility sends him/her the rows necessary to continue the
computation.
\end{enumerate}

As one can observe, the above idea is a kind of non-interactive
spreadsheet, very much like sed is a non-interactive text editor. 

The same approach can be used to execute computations specified by
means of row-directed but not necessarily row-organized
spreadsheets. The fundamental drawback of this class is their much
higher memory consumption during evaluation.

One should keep in mind, however, that the above idea requires
considerable amount of further work to become available in
practice. First of all, EREW PRAM algorithms we demonstrate in this
paper must be suitably adapted to the architectures of existing
parallel or distributed computation systems. Next, the language of
spreadsheets should be enriched by suitable formula shorthands, making
programming easier -- the formulas used in our electronic appendix
spreadsheets S1 and S6 are anything but simple and natural, even
though they implement only bare-bones programming languages. An
alternative might be a tool external to the spreadsheet, translating a
richer language into the formulas composed from the basic
functions. Anyway, in order to be successful, the resulting enriched
language of formulas should be a reasonable compromise between the
original language of spreadsheets, which has very strong data analysis
component and no support for programming constructs, and the PRAM
programming language, which has no data analysis primitives.

\section{Bi-directed spreadsheets: Complexity}\label{S10}

A bi-directed, column organized spreadsheet extended to dimensions $r$
and $c$ can be, according to Theorem 2, evaluated on a PRAM using
$O(cr)$ processors and $O(cr)$ cells of memory and

\begin{enumerate}
\item in time $O(c\log r)$ if treated as column-organized
column-directed.

\item in time $O(r\log cr)$ if treated as row-directed.
\end{enumerate}

One might be tempted to believe that it is possible to combine somehow
those two methods together to yield a parallel evaluation algorithm of
even better time complexity. Optimally, it might seem plausible that
time complexity $O((\log c\log r)^{{O(1)}})$ could be achieved.

However, we prove below that there is no evaluation algorithm of this
complexity unless $ P=NC,$ an unlikely event.

\begin{theorem}\label{T5} There exists a bi-directed,
column-organized spreadsheet $S_{3}$, such that when it is extended to
dimensions of $3n$ rows and $8n$ columns, it computes the solution to
the CVP problem of size $ n$, given its description as input.
\end{theorem}
\begin{proof} The main idea is to implement CVP
``diagonally'', two fully commented implementations are provided in an
  electronic appendix S3. One of them is conceptually simpler, but
  uses two-dimensional \texttt{INDEX} calls. The other is even
  column-organized, but technical complication is the price to pay for
  this structural property.
\end{proof}

Bi-directed spreadsheets are clearly more restrictive than those which
are directed in one dimension only. Author's personal experience from
the development of $S_{3}$ is that the bi-directed structure is
quite unnatural, especially in the column-organized version. Otherwise
very simple computation of CVP required a significant effort to be
programmed. At the same time this structure does not seem to offer any
noticeable advantage in terms of complexity of evaluation.

\section{Un-directed spreadsheets II: What can they compute?}\label{S11}

After a successful implementation of a PRAM in a row-directed
spreadsheet and demonstrating that a large class of PRAM computations
can be expressed in spreadsheets, it seems natural to attempt a
similar goal for un-directed ones, too.

We demonstrate below, that one can create an un-directed row-organized
spreadsheet which implements PRAM in a much more flexible way than the
row-directed row-organized one.

\begin{theorem}\label{T6} There exists a restricted
row-organized (but not row-directed) spreadsheet $S_{2}$ consisting of
21 cells (\texttt{A2} to \texttt{A22}) with formulas using restricted
set of functions, one row for input and a separate area for input
interpreted as a program and for a value of $p$, such that for every
CRCW PRAM $A$ with program $\pi$ for every input vector $ v$, if one

\begin{enumerate}
\item  pastes the encoding of $\pi$ into the program
area of $S_{2}$ 

\item marks and fills the initial range \texttt{A2:A22} to the right
  for $q$ columns,

\item selects rows \texttt{13:22} of these $q$ columns and fills
  downward so that the bottom row of the new range is $10t+12$,

\item  pastes into $S_{2}$  the input vector
$v$ of $A$ in the first row; 

\item  inserts a number into the input cell $p$ 
\end{enumerate}
then the cells at the intersection of the bottom 10 rows with $p$
columns of numbers from $q-p-q(modp)$ to $q-q(modp)$ compute the state
of $A$ after $t*(q/p-1)$ steps of computation on $v$.
\end{theorem}

Informally, the spreadsheet $S_{2}$ above is able to simulate any PRAM
$A$, for which $ p=m$ in such a way, that filled to $ 10t$ rows and
$q\ge p$ columns, it can utilize this computation area to encode a
PRAM with $p$ processors and $p$ cells of shared memory, running for
$tq/p$ time. Moreover, the parameter $p$ is a part of the input, so
only the whole input specifies, how many processors will be used in
the computation. In particular, for $p=1$, this results in a fully
sequential computation of length $ tq.$

\begin{proof} The commented spreadsheet is provided as electronic
  appendix S2. However, it is recommended that the reader first
  analyses S1, on which S2 is based.
\end{proof}

It is instructive to compare $S_{5}$ from Section \ref{S5} with the
present $S_{2}$. It seems at the first glance that the former is a
special case of the latter. However, it is not the case. During the
whole computation expressed in $S_{5}$ it is always possible to refer
to the values computed in the past, no matter how distant. In $S_{2}$
the simulated PRAM can only refer to the values computed in the
previous step of simulation. This indicates the difficulty of
describing the computations of a spreadsheet by a machine model. In a
spreadsheet, every cell is computed exactly once, but its value
remains accessible forever. In typical machine models of computation,
memory locations can be reused (overwritten), but once this is done,
their old values become inaccessible.

\section{Directed spreadsheets: Limitations}\label{S12}

First of all, macros and user-defined functions written in general
purpose programming languages, such as Visual Basic in \textit{Excel\/}
are excluded from the spreadsheets we analyze for the reason that in
their presence all our above considerations become futile. Indeed, in
their presence spreadsheet cells become mere output locations where
the results of absolutely unrestricted computations of the macros and
functions are written. Therefore the structure of the spreadsheet does
not reveal any information about how its results are computed.

Our presentation does not cover all functions of spreadsheets. We
would like to mention a few of them here, explaining the reason why
they are excluded from our analysis.

\texttt{OFFSET} is a function which differs significantly from all
other functions mentioned in this paper. It allows the user to specify
an arbitrary range of cells, whose location is determined by the
numerical arguments of \texttt{OFFSET}. \texttt{INDEX} in its full,
two-dimensional syntax is similar in purpose and function. However,
there is a fundamental difference between them. In the case of
\texttt{OFFSET} calls their arguments are calculated at runtime and
only then it becomes known, what cell each \texttt{OFFSET} call refers
to. In particular, when this function is used, it is impossible to
determine, before executing the spreadsheet, if circular references
are created, if the spreadsheet is row-organized or
row-directed. Generally, all features of spreadsheets which organize
our discussion in this paper become meaningless in the presence of
\texttt{OFFSET}. The full syntax of \texttt{INDEX} to the contrary,
has a two-dimensional range as an argument limiting the area where the
reference is created, and the full dimensions of this range is used to
determine if circular references are created. This check is performed
before running the spreadsheet.

Functions \texttt{ADDRESS} combined with \texttt{INDIRECT} can give a
similar effect, since the former can be used to create a text address
of any cell whose numerical coordinates are passed as arguments, and
the latter transforms that text into the actual reference.

\texttt{SUMIFS} and \texttt{COUNTIFS} are allowed in Section 7 in a
limited form. The complete form, which permits specifying two (or more
rows) with inequality-type criteria and an additional row whose values
should be added, refers to an unsolved problem in the theory of
algorithms. If we consider the values in the first two rows as
coordinates of points in the plane and the values in the third column
as values located in those points, then specifying inequality-type
criteria for summation or counting amounts to summing or counting the
elements in an arbitrary quadrant of the plane. There is currently no
data structure and algorithm known to allow for inserting points with
values and summing values from any quadrant in time logarithmic in the
number of specified points. However, there is no proof that this is
impossible, either.

Circular references break the relation between the computation cost
and the dimensions of the spreadsheet, because every cell is
potentially recomputed many times. Moreover, the principle of
direction becomes difficult to formulate. Electronic appendix S7 is
the implementation of a PRAM in a spreadsheet with circular
references. It is completely analogous to S1, except that it does not
require creating additional rows. Instead, one has to enable circular
references in the \textit{Excel\/} options and change the value of the
red SWITCH cell to 1 to run the computation. Inserting 0 in this cell
resets the simulator to the initial state.

Presence of circular references brings about a completely new
phenomenon in the spreadsheets -- \textit{nondeterminism}. Consider a
three-cell example in Figure 6.

\begin{figure}
\includegraphics[width=3in]{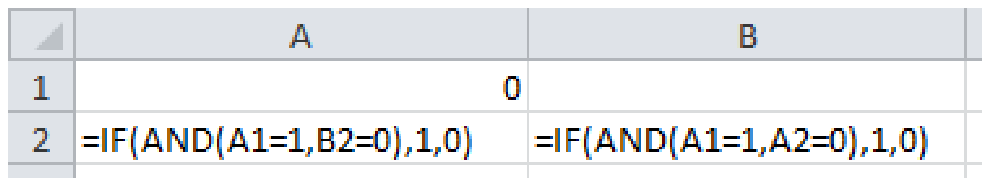}
\caption{An example of nondeterminism in a spreadsheet with circular
  references. If the value of \texttt{A1} is changed to 1, either
  \texttt{A2} or \texttt{B2} changes its value, but it is a
  nondeterministic choice, which one}
\label{F6}
\end{figure}

If in this situation we change the value of \texttt{A1} to 1, both
\texttt{A2} and \texttt{B2} should be recomputed, because both depend
on \texttt{A1}. Their contents are absolutely symmetrical. Which one
will be recomputed first is decided by the spreadsheet cell
recomputation mechanism. Recomputation changes its value from 0 to 1,
and thereby prevents the other cell from changing the value from 0 to
1. Thus the spreadsheet breaks the symmetry between the contents of
\texttt{A2} and \texttt{B2} and no logical examination of the formulas
in the spreadsheet can determine which cell will change its value to 1
and which won't.  Therefore the characterization of the expressive
power of spreadsheets with circular references most likely should
refer to nondeterministic complexity classes.

\section{Summary and related research}\label{S13}

We have turned spreadsheets into a class of algorithms, by assuming
that they are small templates, which are filled to a larger area of
the grid to process data of variable size.

Under this scenario we have identified simple structural properties of
spreadsheets, defined in the terms of the pattern of references
between cells, which determine the complexity of the expressible
computations.

One of the classes of spreadsheets appears to correspond closely to a
natural parallel complexity class. It is therefore a reasonable
candidate to become a general purpose parallel programming language
dedicated for users without programming training.

There was very little previous research on the computational power of
spreadsheets. The papers \cite{bernstein}, \cite{tm} and
\cite{animation} demonstrate simulations of various algorithms and
models o computation using spreadsheets, but without any intent to
estimate the full power of this computation paradigm. Paper
\cite{SIGMOD} demonstrates how to implement relational algebra queries
in spreadsheets (which turn out to be column-organized, but not
necessarily column-oriented).

\bibliographystyle{plain}


\appendix
\section*{Electronic appendices}

Electronic appendices can be downloaded from author's Web page
\\
\url{http://www.mimuw.edu.pl/~jty/Spreadsheets}.\\ 
Generally,
spreadsheets referred to as S$i$ has name
\texttt{Tyszkiewicz\_S$i$.xlsx} there.
\end{document}